%% file: main-characterization.tex
\documentclass[a4paper,UKenglish,cleveref, autoref, thm-restate]{lipics-v2021}

\usepackage[utf8]{inputenc}
\usepackage[T1]{fontenc}
\usepackage{lmodern}
\usepackage{paralist}

\usepackage{tikz}

\usetikzlibrary{automata}
\usetikzlibrary {arrows.meta,automata,positioning}

\usepackage{amsmath} 
\usepackage{amssymb}
\usepackage{mathrsfs}
\usepackage{amsthm}
\usepackage{amsfonts}

\usepackage{graphicx}
\usepackage{algorithm}
\usepackage{algorithmic}
\usepackage{hyperref}
\usepackage{stmaryrd}
\usepackage{mathtools}
\usepackage{todonotes}
\hypersetup{hidelinks,breaklinks}
\usepackage{cleveref}
\usepackage{thm-restate}

\usepackage[normalem]{ulem}

\usepackage[notion, quotation, electronic]{knowledge}

\knowledgestyle{notion}{color=darkgray}

\input{pieces_of_knowledge.default.kl}
\input{macros.tex}

\hideLIPIcs
\nolinenumbers

\knowledgenewrobustcmd\Out{\cmdkl{\mathrm{Out}}}
\knowledgenewrobustcmd\In{\cmdkl{\mathrm{In}}}

\knowledgenewcommand{\Lang}{\cmdkl{\L}}

\knowledgenewcommand{\cc}{\mathbin{\cmdkl{\cdot}}}

\knowledgenewrobustcmd\skel{\cmdkl{\mathrm{skel}}}

\knowledgenewcommand\LangS{\cmdkl{\L}}

\knowledgenewcommand{\Ai}{\cmdkl{A_D}}
\knowledgenewcommand{\De}{\cmdkl{\Delta}}

\knowledgenewrobustcmd\best[1][t]{\cmdkl{\mathrm{best}_{#1}}}
\knowledgenewrobustcmd\leader[1][t]{\cmdkl{\mathrm{leader}_{#1}}}

\knowledgenewrobustcmd\stateleq{\mathrel{\cmdkl{\sqsubseteq}}}
\newrobustcmd\stategeq{\mathrel{\kl[\stateleq]{\sqsupseteq}}}
\newrobustcmd\statel{\mathrel{\kl[\stateleq]{\sqsubset}}}
\newrobustcmd\stateg{\mathrel{\kl[\stateleq]{\sqsupset}}}

\knowledgenewrobustcmd\transleq{\mathrel{\cmdkl{\leqslant}}}
\knowledgenewrobustcmd\transgeq{\mathrel{\kl[\transleq]{\geqslant}}}
\knowledgenewrobustcmd\translt{\mathrel{\kl[\transleq]{<}}}

\knowledgenewcommand{\Si}{\cmdkl{\Sigma}}
\knowledgenewcommand{\ti}{\cmdkl{t_{\stateleq}}}
\knowledgenewcommand{\Sii}{\cmdkl{\Sigma_{\stateleq}}}
\knowledgenewcommand{\unit}[1][\Sigma]{\cmdkl{1_{#1}}}

\knowledgenewcommand{\ska}[1][t]{\cmdkl{\delta_{#1}}}
\knowledgenewrobustcmd\red{\cmdkl{\mathit{red}}}
\knowledgenewrobustcmd\green{\cmdkl{\mathit{green}}}

\knowledgenewcommand{\io}{\cmdkl{\iota}}

\knowledgenewrobustcmd\inter[1]{\cmdkl{\mathit{[#1]}}}

\knowledgenewrobustcmd\eps{\cmdkl{\varepsilon}}
\knowledgenewrobustcmd\Pieps{\cmdkl{\Pi_\varepsilon}}

\knowledgenewcommand\resleqL{\mathrel{\cmdkl{\leqslant_L}}}

\title{An algebraic characterisation of Eve-positional languages}

\titlerunning{An algebraic characterisation of Eve-positional languages} 

\author{Thomas Colcombet}{Université Paris Cité, CNRS, IRIF, France}{thomas.colcombet@irif.fr}{https://orcid.org/0000-0001-6529-6963}{}

\author{Olivier Idir}{Université Paris Cité, CNRS, IRIF, France}{olivier.idir@ens-lyon.org}{https://orcid.org/0009-0003-3848-8515}{}

\authorrunning{T. Colcombet and O. Idir} 

\Copyright{Thomas Colcombet and Olivier Idir} 



\acknowledgements{We thank Antonio Casares and Pierre Ohlmann for their enthusiastic answers concerning their article.}

\begin{document}
	\maketitle

	\begin{abstract}
		We present a new algebraic characterisation of Eve-positionality for $\omega$-regular languages.
		It involves only a limited number of elementary local properties to be checked.
		
		An $\omega$-regular language is Eve-positional if, in all games with this language as objective, the existential player  (Eve) can play optimally without keeping any information concerning the history of the moves seen so far.
		This notion plays a crucial role in verification, automata theory and synthesis.
		
		Our proof heavily relies on a recent result of Casares and Ohlmann which states several characterisations of Eve-positionality for $\omega$-regular languages.
		More precisely, it relies on their a 1-to-2 player lift result:  for an $\omega$-regular language, being Eve-positionally over all finite Eve-only arenas suffices for being Eve-positional over all two-player arenas.	
	\end{abstract}

	\section{Introduction}
	\label{section:introduction}
	\input{intro-c.tex}

	\section{General definitions}\label{sec:defs}
	\input{defs-c.tex}

	\section{Characterisation of Eve-positional languages}\label{sec:charac}
	\input{characterization.tex}
	
	\bibliography{biblio}
\end{document}

%% file: macros.tex
\newcommand{\NN}{\mathbb{N}}

\newcommand{\cP}{\mathcal{P}}

\newcommand{\G}{\mathcal{G}}

\renewcommand{\L}{\mathcal{L}}

\renewcommand{\S}{\mathcal{S}}

\knowledgenewcommand{\brackets}[1]{\cmdkl{\mathit{[#1]}}}
\knowledgenewcommand{\floor}[1]{\left \lfloor #1 \right \rfloor}
\knowledgenewcommand{\ceil}[1]{\left\lceil #1 \right\rceil}
\knowledgenewcommand{\restricted}[1]{_{\cmdkl{\mid}  #1 }}
\knowledgenewrobustcmd\dom{\cmdkl{\mathrm{dom}}}
\knowledgenewrobustcmd\img{\cmdkl{\mathrm{img}}}
\knowledgenewrobustcmd\invImage[1]{^{\cmdkl{-1}}(#1)}
\knowledgenewrobustcmd\pref[1]{\cmdkl{_{\mid #1}}}
\knowledgenewrobustcmd\trans[1]{\mathrel{\cmdkl{\xrightarrow{#1}}}}
\knowledgenewrobustcmd\transpath[1]{\mathrel{\cmdkl{\xrightarrow{#1}}}}
\newrobustcmd\transsub[2]{\mathrel{\kl[\trans]{\xrightarrow{#1}_{#2}}}}
\knowledgenewrobustcmd\ptrans[1]{\mathrel{\cmdkl{\overset{#1}\rightsquigarrow}}}

%% file: intro-c.tex

\subsection*{Context}
The problem of reactive synthesis is the following: given a system equipped with a formal specification and interacting with its environment for a possibly infinite amount of time, we want to design a controller that can react to the environment while ensuring that its specification is met. 

A formalism frequently used for this is the one of games on graphs. In these games, two players, Eve and Adam, move a token along the edges of an edge-coloured directed graph called an arena. Its vertices are partitioned between those belonging to Eve and those belonging to Adam, and when the token arrives in a vertex, its owner chooses the next edge, which the token will follow. This interaction continues infinitely, thus producing an infinite path through the arena and a corresponding infinite sequence of colours $w$ (also called an $\omega$-word). A play is won by Eve if the corresponding word belongs to some set $W$, described as her objective. Else, it is won by Adam. Over the games we consider, it always stands that one of the players has a winning strategy, which describes a way to win the game even when this strategy is known in advance by the opponent.

The interaction between the system and its environment can be modeled by such a game, where the specification provides the winning objective $W$ \cite{Buchi_Landweber_synthesis}. In such a context, the complexity of the player's winning strategy is crucial, as games with simple winning strategies are generally easier to solve algorithmically and result in controllers that can be represented more succinctly.

\noindent \textbf{$\omega$-regular languages.}
An "$\omega$-regular language" is a set of "$\omega$-words" that can be described by many equivalent formalisms (e.g., automata that can be either deterministic, non-deterministic, or alternating; one-way or two-way; with Muller, Parity, Rabin, or Streett acceptance conditions; or by monadic-second-order logic). The class of "$\omega$-regular languages" plays a pivotal role in many decision procedures and mathematical arguments related to verification, automata theory, controller synthesis, etc.

\noindent \textbf{Eve-positionality.}
One of the simplest cases, when attempting to solve a game, is when the winning strategy is positional. This means that the corresponding player does not need any memory of the past to play optimally: they can make their decision simply looking at the current state of the game. Note that this memoryless property need not be the same for both players. In this work, we are considering the objectives $W$ such that over all graphs, if Eve wins, she can do so with a positional strategy. They are called "Eve-positional languages".\\
Until recently, little was known about this class. Major progress was made in the last few years, but there still lacked a convenient formalism to handle them, which we propose here.

\subsection*{Contribution}
\AP In this paper, we propose a new algebraic characterisation of "$\omega$-regular" "Eve-positional languages".
This characterization does not rely on an automaton recognising the language, but describes the language's structure directly.
\AP It consists of three conditions, called the ""local preference properties"", each one of the form ``if $u_1 \in L$ (and possibly $u_2 \in L$), either $v_1\in L$ or $v_2\in L$'':
\begin{restatable}{theorem}{EveposiLocal}\label{thm:local-posi}\AP
	An "$\omega$-regular language" ~$L\subseteq C^\omega$ is "Eve-positional" if and only if, for all $u, u'\in C^*$, $v,v'\in C^+$, and $w,w'\in C^\omega$, the following "local preference properties" hold:
	\begin{asparaenum}
		\item if $uw\in L$ and $u'w'\in L$, then $uw' \in L$ or $u'w \in L$,
		\item if $uvw \in L$, then $uv^\omega \in L$ or $uw \in L$, and
		\item if $u(vv')^\omega \in L$, then $uv^\omega \in L$ or $uv'^\omega \in L$.
	\end{asparaenum}
\end{restatable}
The fact that the "local preference properties" are necessary is easy to establish.  The proof of the other --- difficult --- direction makes crucial use of two results concerning $\omega$-regular games.
The first one is the classical result of finite memory-determinacy of $\omega$-regular games.
\begin{theorem}[\cite{Buchi_Landweber_synthesis}]\AP\label{theorem:finite-memory-determinacy}
	"Games" with "$\omega$-regular" conditions are "finite-memory" "determined".
\end{theorem}
The second one is one of the recent deep results of Casares and Ohlmann that characterise one-player positional determinacy for $\omega$-regular conditions:
\begin{theorem}[Theorem 3.1 in \cite{CO24Positional}, $1$-to-$2$ player lift]\AP\label{theorem:CO}
	An "$\omega$-regular language" which is "uniformly Eve-positional" over all finite "Eve-only arenas" is "Eve-positional" over all arenas.
\end{theorem}
Note that in \cite{CO24Positional}, the definition of positionallity used is in fact implicitly uniform positionality.
"Uniform Eve-positionality" is stronger than "Eve-positionality": it means that a single positional strategy is sufficient for winning from all the vertices in the winning region.

\subsection*{Related works}

\noindent \textbf{Eve-positionality}
The "Eve-positionality" of some central classes was established as early as the 90s, such as parity languages \cite{Emerson_Jutla_parity} and Rabin languages \cite{RabinPosi}.
The first thorough analysis of "Eve-positionality" was undertaken by Kopcziński (see \cite{Kopczynski06}).
He notably established that it is decidable whether a prefix-independent "$\omega$-regular language" is "Eve-positional"  \cite{Kopczynski06,Kopczynski07}, albeit with a quite slow procedure (in $O(n^{O(n^2)})$), and is not very informative on the reason why the language is "Eve-positional".
More generally, despite his breakthroughs, he did not obtain any general characterisation of "Eve-positional" languages.

Some partial characterisations were established in the following decade on specific subclasses of "$\omega$-regular languages". Colcombet, Fijalkow, and Horn established a semantic characterisation of "Eve-positionality" for safety languages \cite{CFK_safety_posi}, that is, languages where it is known in finite time whether some word is out of the language.
A characterisation for a larger class was actually obtained a few years later by Bouyer, Casares, Randour, and Vandenhove \cite{Bouyer_Buchi_posi}, as they obtained a semantic characterisation of "Eve-positional" languages recognised by deterministic Büchi automata. This, however, does not correspond to the general $\omega$-regular case, and for instance fails to cover the deterministic coBüchi case.

Recently, using the well-monotonic universal graphs initially developed by Ohlmann \cite{Ohlmann_universal_graphs}, he and Casares gave several other characterisations for "Eve-positionality" \cite{CO24Positional}.
They also give a polynomial algorithm which, given a deterministic parity automaton as input, decides whether the language it defines is "Eve-positional".
They notably established that an "$\omega$-regular language" $L$ is "Eve-positional" over all arenas (potentially infinite and containing $\varepsilon$-moves) if and only if $L$ is "Eve-positional" over finite $\eps$-free Eve-only arenas (see \Cref{theorem:CO} above).
These articles provide the first characterisation achieved over all "$\omega$-regular" languages. However, they require one to reason with an automaton recognising the language, rather than directly with the language and inclusion properties.

\subsection*{Structure of the paper}

Necessary definitions are presented in \Cref{sec:defs}. \Cref{thm:local-posi}  is proved in \Cref{sec:charac}.

%% file: defs-c.tex

\subsection{Words and languages}
\label{subsection:words}

For a set A, its powerset is denoted $\cP(A)$.

\AP An ""alphabet"" is a finite set $C$, whose elements are called letters.
The sets $C^*$ and $C^\omega$ respectively denote the sets of finite ""words"" and infinite \reintro*\kl{words} of length $\omega$ over $C$ (also denoted ""$\omega$-words""). Subsets of $C^*$ and $C^\omega$ are called languages. 

In this article, we are interested in ""$\omega$-regular languages"", which can be defined in many equivalent ways. They are notably the languages that can be recognized by automata over infinite words.

An "$\omega$-word" is ""ultimately-periodic"" if it is of the shape $u v^\omega$, with $u$ and $v$ finite words, and $v$ is non-empty.
The set of non-empty finite words is denoted $C^+$, and the ""empty word"" is denoted $\intro*\eps$.

\AP The length of a finite "word" $u$ is written $|u|$. For a (possibly infinite) word $u$ and $n\in \NN$, $u_n$ denotes its letter at index $n$.

For a function $f: A \to B$ and $b\in B$, we denote by $f\intro*\invImage b $ the set $\{a \in A \mid f(a) = b\}$.

\subsection{Graphs and morphisms}
\label{subsection:graphs}

\AP \textbf{Graphs.} A ""directed labelled graph"" is a tuple $G = (V, C, E)$ (or simply a \reintro*"$C$-graph" or a \reintro*"graph" if there is no ambiguity), where $V$ is a countable set of ""vertices"", $C$ is an "alphabet" of ""labels"", and $E \subseteq V \times C \times V$ is a set of ""edges"". It is said ""finite@@graph"" if both $V$ and $E$ are finite.

\AP The ""$c$-edge from vertex $v$ to vertex $v'$"", if it exists, is $(v,c,v')\in E$. We sometimes denote the existence of such an "edge" by $v \trans{c} v'$, possibly omitting $c$ if it is not relevant.
For a given $v\in V$, edges of the form $v\intro*\trans{}v'$ are called the outgoing edges of $v$, and the set of its outgoing edges is denoted $\intro*\Out(v)$. A ""sink"" is a "vertex"~$v$ such that $\Out(v) = \emptyset$. A "graph" is ""sinkless"" if it does not contain any "sink". By default, all the graphs we consider are "sinkless".

\AP \textbf{Paths.} A ""path"" in $G$ is a (possibly infinite) sequence of edges $\pi = (v_i, c_i, v_{i+1})_{i < \alpha}$ with matching endpoints. We say that $\pi$ starts in $v_0$, visits the vertices ${v_i \mid i < \alpha}$, and has for ""label@path"" the word $u = (c_i)_{i< \alpha}$. 
If $\alpha$ is finite, we say that $\pi$ is from $v_0$ to $v_\alpha$, and denote the existence of a finite "path" by $v_0 \intro*\transpath{u} v_\alpha$. We can similarly omit $u$ if it is not relevant.

\AP \textbf{Morphisms.} Given two "$C$-graphs" $G = (V,C,E)$ and $G' = (V',C,E')$, a ""morphism@@graph"" from $G$ to $G'$ is a map $\phi:V \to V'$ such that
\[
v \transsub c G v'\ \text{ implies }\ \phi(v) \transsub c {G'} \phi(v').
\]
\AP If such a "morphism" exists, we say that $G$ ""maps into"" $G'$, or that $G'$ \reintro*"embeds" $G$. Note that this property extends over "paths", therefore each "path" in $G$ can then similarly be mapped to a corresponding "path" in $G'$ of same "label@@path".

\subsection{Games and strategies}
\label{subsection:games}

We describe here the model of games that we use: infinite duration game with perfect information between two antagonistic players (""Eve"" and ""Adam""), played on a graph.

\AP A ""$W$-game"" (or simply a \reintro*"game" when there is no ambiguity)  is a tuple $\G = (V_E, V_A, C, E, W)$ in which
\begin{itemize}
	\item $V_E$ and $V_A$ are disjoint sets of vertices; the vertices in $V_E$ are ""controlled by Eve"" and the ones in $V_A$ are ""controlled by Adam"";
		we set $V = V_E \uplus V_A$
	\item $G:= (V, C, E)$ is a "sinkless graph" called its ""arena"", and
	\item $W\subseteq C^\omega$ is a ""winning condition"" (for Eve).
\end{itemize}
\AP A "game" is ""Eve-only"" (resp. \reintro*"Adam-only") if $V_A = \emptyset$ (resp. $V_E=\emptyset$). 
\AP Given some vertex~$v_0$, a ""play from~$v_0$"" is an infinite path in $G$ that starts in vertex~$v_0$.
The ""label of the play"" is the word in $C^\omega$ obtained by projecting out the states and keeping the "label" of "edges" only.
A "play" is ""winning@@play"" if the "its label@label@play" belongs to $W$.

\AP Intuitively, a "game" is played as follows: starting from some "vertex"~$v_0$, successive edges, called ""moves"", are played, where the player "controlling" the current vertex $v$ chooses a "move" $(v,c,v') \in E$, and the process continues from vertex $v'$. 
This interaction continues forever, producing a "play" $\rho$. The goal of "Eve" is to guarantee that this "play" is "winning@@play" (whatever is the choice of the opponent). If this is the case, "Eve wins from~$v_0$". This intuition is made formal using the notion of "strategy".

\begin{example}
It is customary to denote the "vertices controlled by Eve" as circles, and the ones "controlled by Adam" as squares.
"Moves" are represented by edges, annotated by their label. Consider for instance the following "$W$-game", where $W$ is the set of words over $\{a,b,c,d\}$ such that either there are finitely many $a$'s, or there are infinitely many $b$'s and infinitely many $c$'s:
\begin{center}
\begin{tikzpicture}[baseline={(0,0)},node distance=1cm, empty/.style={}, player0/.style={circle, draw=blue!60, fill=blue!20, minimum size=.3cm}, player1/.style={rectangle, draw=blue!60, fill=blue!20, minimum size=.3cm},>=Stealth]
  \node[empty] (i) at (.4,0) {};
  \node[player1] (u) at (1,0) {};
  \node[player0] (center) at (2.5,0) {};
  \node[empty] (e0) at (3.5,.5) {};
  \node[empty] (e1) at (3.5,-.5) {};
  \path[->] (u) edge [bend left] node [above]  {$a$} (center);
  \path[->] (center) edge [bend left] node [below]  {$d$} (u);

  \tikzset{every loop/.style={min distance=0.8mm,in=140,out=60,looseness=8}}
  \path[->] (u) edge [loop above] node               {$c$} ();

  \tikzset{every loop/.style={min distance=0.8mm,in=-140,out=-60,looseness=8}}
  \path[->] (u) edge [loop below] node               {$b$} ();

  \tikzset{every loop/.style={min distance=0.8mm,in=40,out=120,looseness=8}}
  \path[->] (center) edge [loop above] node               {$c$} ();

  \tikzset{every loop/.style={min distance=0.8mm,in=-40,out=-120,looseness=8}}
  \path[->] (center) edge [loop below] node               {$b$} ();

\end{tikzpicture}
\end{center}
In this game, "Eve wins" from both vertices. One way to achieve this is as follows: as soon as the play reaches the right vertex ("controlled by Eve") Eve plays the move labelled $c$, then the one labelled $b$. This is a winning strategy since whatever Adam decides to do, either he stays ultimately in the left vertex, and only finitely many $a$'s will have been produced (hence the "play is winning"), or it visits infinitely many times the right vertex, and the strategy of Eve guarantees that infinitely many $b$'s and infinitely many $c$'s are encountered (again, the "play is winning").
\end{example}


\AP \textbf{Strategies.} Let $\G = (V_E,V_A, C, E, W)$ be a "game". An ""Eve-strategy"" is a tuple $\S = (S, \sigma)$ where $S$ is a "sinkless graph" and $\sigma$ is a "morphism@@graph" from $S$ to $G$ such that for all "vertices"~$s$ of~$S$, 
	\begin{itemize}
		\item if $\sigma(s)$ is "controlled by Adam", then for all "moves" of the form $(\sigma(s),c,v')\in E$, there is an edge of the form $(s,c,s')$ in~$S$ such that $\sigma(s')=v'$;
		\item if $\sigma(s)$ is "controlled by Eve", then there exists exactly one edge of the form $(s,c,s')$ in~$S$.
	\end{itemize}
	
Note that the morphism~$\sigma$ sends infinite "paths" in~$S$ to "plays" in~$G$.

\AP For $s_0$ a "vertex of~$S$", the "strategy"~$\S$ is ""winning@@strategy"" from~$s_0$ if all "paths" in~$S$ from~$s_0$ are sent to "winning plays" by $\sigma$.
""Eve wins from~$v_0$"" if there exists a "winning strategy from~$s_0$" with $\sigma(s_0)=v_0$.
\AP The ""winning region (for Eve)"" is the set of "vertices~$v$" of~$G$ such that "Eve wins from~$v$".
\AP A game is said ""determined"" if the union of the "winning regions" of the two players covers the whole "arena". Said differently, for all vertices~$v$, one of the two players "wins from $v$".

\AP \textbf{Memory and positionality.}
\AP A "strategy"~$\S = (S, \sigma)$ is of ""finite memory"" if $\sigma^{-1}(v)$ is finite for all vertices~$v$.
The ""memory of the strategy"" is the least natural number~$k$ (if it exists) such that $|\sigma^{-1}(v)|\leqslant k$ for all "vertices"~$v$.
The "strategy" is ""positional@@strategy"" if it is of "memory"~$1$.

\AP Given a vertex~$v_0$, ""Eve wins positionally from~$v_0$"" if there exists a "positional strategy" $(\S,\sigma)$ which is "winning from~$s_0$" for some (the only) $s_0\in\sigma^{-1}(v_0)$.
\AP A "game" is ""uniformly positionally won by Eve"" if there exists a "positional strategy" $(S, \sigma)$ such that for all $v$ in the "winning region for Eve", $\S$ "wins from~$v$" for some (the only) $v\in\sigma^{-1}(s)$.

\AP Given a language $L\subseteq C^\omega$, it is ""Eve-positional"" over a class of games if, for all games in the class, and all "vertices"~$v$, if "Eve wins from~$v$" then "Eve wins positionally from~$v$" for a ""uniformly Eve-positional"" over a class games, if all the "$L$-games" in the class are "uniformly positionally won by Eve".
We shall use this definition for the class of finite "Eve-only games".

The \Cref{theorem:finite-memory-determinacy} thus says that not only are games with "$\omega$-regular" "winning conditions" "determined", but the corresponding "winning strategies" can be chosen to have "finite memory".

\begin{remark}[Eve-strategies induce an Adam-only game]\label{rem:Adam-only-game}
	We observe that an "Eve-strategy" $\S = (S, \sigma)$ can be seen as a "Adam-only game"~$G'$ by setting all the vertices to be "controlled by Adam".
	In this case, the strategy $\S$ is "winning from~$s$" if and only if all the infinite "paths from~$s$" in~$\S$ have their "label@@path" in~$W$, if and only if "Eve wins from $s$ in~$G'$".	
\end{remark}

%% file: characterization.tex

The objective of this section is to prove \Cref{thm:local-posi}, that we recall now:
\EveposiLocal*

Le us first give an example of a language that is not known to be "Eve-positional" from the literature (it is not a safety language, nor a Rabin condition), and show how to use \Cref{thm:local-posi} for proving it.
\begin{example}
	
	Consider, over the alphabet $C = \{a,b,c\}$, the language of  infinite words that
		``contain infinitely many occurrences of the factor $ac^*a$ and finitely many occurrences of the factor $bc^*b$.''
	Using the syntax of $\omega$-regular expressions, this could be expressed at
	\[L=C^*\,(a\,c^*(b\,c^*a\, c^*)^*)^\omega\ .\] 
	Let us show, using \Cref{thm:local-posi} that $L$ is "Eve-positional": we have to prove that it satisfies the three "local preference properties". 
	 
	 \AP Note first that $L$ is ""prefix independent"", which means that $u\in L$ if and only if $vu\in L$ for all~$v\in u^*$.
	 As a consequence, the two first "local preference properties"  immediately hold.
	 	
	For the third "local preference property", since $L$ is "prefix independent", we have show that for all $v,v'\in C^+$ such that $(vv')^\omega \in L$, then either $v^\omega\in L$ or ${v'}^\omega\in L$. 
	Let us fix ourselves two such words $v,v'\in C^+$ such that $(vv')^\omega\in L$. 
	We proceed by case distinction.
	
	Consider first the case $v'\in C^+$ (resp. $v\in C^+$), then $(vv')^\omega\in L$ if and only if $v^\omega\in L$ (resp. ${v'}^\omega\in L$). Hence the third "local preference property" holds.
	
	Otherwise, let $\delta_1$ be the first non-$c$ letter occurring in $v$, $\delta_2$ be the last one, and let $\delta'_1$ and $\delta'_2$ be defined similarly from~$v'$. 
	Let us remark that if $\delta_1=b$ then $\delta'_2=a$. Indeed, otherwise $v'v$ would contain a factor $bc^*b$, and this would contradict $(vv')^\omega\in L$. In the same way, $\delta_2=b$ implies $\delta'_1=a$, $\delta'_1=b$ implies $\delta_2=a$ and $\delta'_2=b$ implies $\delta_1=a$. The remaining cases are the following:
	\begin{itemize}
	\item If $\delta_1=\delta_2=a$ (or symmetrically if $\delta'_1=\delta'_2=a$), then $v^{\omega}$ contains infinitely many occurrences of $ac^*a$. Furthermore, since $(vv')^\omega$ does not contain the factor $bc^*b$, $v^\omega$ does not either. The third "local preference property" is established.
	\item If $\delta_1=\delta_2=b$, then from the above remarks $\delta'_1=\delta'_2=a$. The previous case applies.
	\item If $\delta_1=a$ and $\delta_2=b$ (or symmetrically if $\delta'_1=a$ and $\delta'_2=b$). Then $\delta'_1=a$. If one would have $\delta'_2=a$, this would be treated in the first case. Hence $\delta'_1=b$.
	In this case, since $ac^*a$ occurs in $(vv')^\omega$, it has to occur either in~$v$ or in~$v'$. In the first case, this means that $v^\omega$ contains infinitely many occurrences of~$ab^*a$, and no occurrences of $bc^*b$, and the third "local preference property" is established. The argument identical for $v'$.
	\item If $\delta_1=b$ and $\delta_2=a$ (or symmetrically if $\delta'_1=b$ and $\delta'_2=a$), the previous argument also holds.
	\end{itemize}

	We thus have proved that $L$ satisfies the thhree "local preference properties", and hence it is "Eve-positional" by \Cref{thm:local-posi}.
\end{example}

We now have to esablish \Cref{thm:local-posi}.
We first show below, in \Cref{lemma:posi-union-gen-easy}, that the "local preference properties" form a necessary condition for a language to be "Eve-positional".
We then establish in \Cref{corollary:difficult-direction} the other direction, ie that if $L$ satisfies the "local preference properties", then it is "Eve-positional" over all arenas.

\begin{lemma}\label{lemma:posi-union-gen-easy}
	If $L\subseteq C^\omega$ is "Eve-positional" over all "finite@@graph", "Eve-only" "$L$-games", then $L$ satisfies the "local preference properties".
\end{lemma}
\begin{proof}
	In the following, for $u\in  C^*$, we denote by $G_u$ the finite "$ C$-graph" consisting of a single thread of vertices, where the only path has label $u$. In a similar way, for $w\in C^\omega$ an "ultimately-periodic" word, we denote $G_w$ some finite "$ C$-graph" that has a single infinite path, labelled by~$w$.
	Note that as $L$ is "$\omega$-regular", it is uniquely characterised by its "ultimately-periodic" words, and we can thus suppose that the $w,w'$ implied in the "local preference properties" are ultimately periodic in this proof.
	For $u,u'\in C^*$, $v,v'\in C^+$ and "ultimately-periodic" $w,w'\in C^\omega$, we consider the following three "$L$-games": \\
	
	\begin{center}
\begin{tikzpicture}[baseline={(0,0)},node distance=1cm, empty/.style={}, player0/.style={circle, draw=blue!60, fill=blue!20, minimum size=.3cm}, player1/.style={rectangle, draw=blue!60, fill=blue!20, minimum size=.3cm},>=Stealth]
  \node[empty] (i) at (.4,0) {};
  \node[player1] (u) at (1,0) {};
  \node[player0] (center) at (2.5,0) {};
  \node[empty] (e0) at (3.5,.5) {};
  \node[empty] (e1) at (3.5,-.5) {};
  \path[->] (i) edge   (u);
  \path[->] (u) edge [bend left] node [above]  {$u$} (center);
  \path[->] (u) edge [bend right] node [below]  {$u'$} (center);
  \path[->] (center) edge [bend left] node [above ]  {$w$} (e0);
  \path[->] (center) edge [bend right] node [above ]  {$w'$} (e1);
\end{tikzpicture}
\quad
\begin{tikzpicture}[baseline={(0,0)},node distance=1cm,  empty/.style={}, player0/.style={circle, draw=blue!60, fill=blue!20, minimum size=.3cm}, player1/.style={rectangle, draw=blue!60, fill=blue!20, minimum size=.3cm},>=Stealth]
\node[empty] (i) at (.4,0) {};
\node[player1] (v0) at (1,0){};
\node[player0] (center) at (2,0)  {};
\node[empty] (e) at (3,0)  {};
\path[->] (i) edge (v0);
\path[->] (v0) edge  node [above ]  {$u$} (center);
\tikzset{every loop/.style={min distance=0.8mm,in=0,in=125,out=50,looseness=8}}
\path[->] (center) edge [loop above] node  {$v$} ();
\path[->] (center) edge node  [above ]  {$w$} (e);
\end{tikzpicture}
\quad
\begin{tikzpicture}[baseline={(0,0)},node distance=1cm, empty/.style={}, player0/.style={circle, draw=blue!60, fill=blue!20, minimum size=.3cm}, player1/.style={rectangle, draw=blue!60, fill=blue!20, minimum size=.3cm},>=Stealth]
  \node[empty] (i) at (.4,0) {};
  \node[player1] (v0) at (1,0){};
  \node[player0] (v1) at (2,0)  {};
  \path[->] (i) edge   (v0);
  \path[->] (v0) edge  node [above ]  {$u$} (v1);
  \tikzset{every loop/.style={min distance=0.8mm,in=0,in=40,out=120,looseness=8}}
  \path[->] (v1) edge [loop above] node               {$v$} ();
  \tikzset{every loop/.style={min distance=0.8mm,in=0,in=-40,out=-120,looseness=8}}
  \path[->] (v1) edge [loop below] node               {$v'$} ();
\end{tikzpicture}
\end{center}

	
	Assume first that $uw\in L$ and $u'w'\in L$. This means that Eve wins the first game, whatever the first "move" chosen by Adam. Since $L$ is "Eve-positional", she does so "positionally@@strat". This means that either she wins by always choosing the up move from~$q$, and in this case $u'w\in L$, or she wins by always choosing the down move from~$q$, and in this case $uw'\in L$. Overall the first "local preference property" is established.
	
	Assume now that $uvw\in L$. This means that Eve wins the second game by going up the first time~$q$ is met, and then down the second time. Since $L$ is "Eve-positional", she wins "positionally@@strat". If this strategy consists in always choosing the up move from~$q$, we get $uv^\omega\in L$. If it consists in always choosing the down move from~$q$, we get $uw\in L$. Overall we have proved the second "local preference property".
	
	Finally, assume that $u(vv')^\omega\in L$. This means that Eve wins the third game by playing alternatively the up move from~$q$, then the downward move, etc.  Since $L$ is "Eve-positional", she does so "positionally@@strat". If she wins by always choosing the up move from~$q$, we get that $uv^\omega\in L$. If it consists in always choosing the down move from~$q$, we get $uv'^\omega\in L$. Overall the third "local preference property" is established.
\end{proof}

Let us now establish the other direction of \Cref{thm:local-posi}, ie that the "local preference properties" suffice for guaranteeing that an $\omega$-regular language is "Eve-positional".
We begin with some elementary remarks on the order $\resleqL$ over infinite words induced by the language $L$, that we introduce now.

\AP Let us denote $w \intro*\resleqL w'$ for $w,w'\in  C^\omega$ holds if  $uw\in L$ implies $uw'\in L$ for all~$u\in C^*$.

Note first that $w\resleqL w'$ implies $uw\resleqL uw'$ for all $u\in C^*$.
\begin{lemma}\label{lemma:leqslant}
	If~$L$ satisfies the "local preference properties", then:
	\begin{itemize}
		\item $\resleqL$ is a total preorder on~$ C^\omega$,
		\item either $vw\resleqL v^\omega$ or $vw\resleqL w$ for all~$v\in C^+$ and $w\in C^\omega$, and
		\item either $(vv')^\omega\resleqL v^\omega$ or  $(vv')^\omega\resleqL {v'}^\omega$, for all~$v,v'\in C^+$.
	\end{itemize}
\end{lemma}
\begin{proof}
	First item. Assume that $\resleqL$ would not be a total preorder, this would mean that there exist $w,w'\in C^\omega$ such that $w \not\resleqL w'$ and $w'\not\resleqL w$. This would mean that there exists~$u,u'\in C^*$ such that $uw\in L$, $uw'\not\in L$, $u'w'\in L$ and $u'w\not\in L$. This contradicts the first "local preference property".
	
	Second item. Assume the conclusion would not hold, this would mean that there exist $u,u'\in C^*$ such that $uvw\in L$ yet $uv^\omega\not\in L$, and $u'vw\in L$ yet $u'w \not\in L$. Since $\resleqL$ is total, there are two cases: either $v^\omega \resleqL w$ or $w \resleqL v^\omega$. If $v^\omega \resleqL w$, we have that $u'v^\omega\not\in L$, and this contradicts the second "local preference property". If $w \resleqL v^\omega$, we have that $uw\not\in L$, and this contradicts again the second "local preference property".
	
	Third item. Assume the conclusion would not hold, this would mean that there exist $u,u'\in C^*$ such that $u(vv')^\omega\in L$, $uv^\omega\not\in L$, $u'(vv')^\omega\in L$ and $u'{v'}^\omega\not\in L$.
	According to the first item, either ${v'}^\omega \resleqL v^\omega$ or $v^\omega \resleqL {v'}^\omega$.
	In the first case, since $u{v}^\omega\not\in L$, we get $u{v'}^\omega\not\in L$.
	Overall, we have  $u(vv')^\omega\in L$, $uv^\omega\not\in L$ and $u{v'}^\omega\not\in L$, which is a contradiction to the third "local preference property". 
	The other case is symmetrical.
\end{proof}

\begin{lemma}\label{lemma:posi-union-gen-hard}
	If $L\subseteq  C^\omega$ is "$\omega$-regular" and satisfies the "local preference properties" then $L$ is "uniformly Eve-positional" over all "finite@graph" "Eve-only games".
\end{lemma}
\begin{proof}
	Let $G$ be an "$L$-game" played on a finite, Eve-only arena won by Eve, and $R$ be the "winning region" for Eve in~$L$.
	We have to show that there exists a "winning@@strategy" "positional strategy" for Eve that is winning from all positions in $R$.
	The key argument is to show that for all finite strategy~$\S=(S,\sigma)$ winning for Eve from~$R$, either $\S$ is positional, or there exists a smaller strategy for Eve (smaller meaning with less memory states) that is also winning from~$R$. Indeed,  this means that if one considers a smallest (in number of memory states) strategy~$\S$ winning for Eve from~$R$, it is positional.
	
	\AP In this proof, an ""initial path"" in the graph $S$ of a "strategy" $\S = (S,\sigma)$ is a path in~$S$, whose starting point is a vertex $v$ such that $\sigma(v)\in R$. Observe that $\S$ is "winning@@strategy" over $R$ if and only if all its "initial paths" are "winning@@play".
	
	Let us consider a finite "strategy"~$\S = (S,\sigma)$ "winning for Eve from@winning strategy"~$R$. Such a strategy exists by \Cref{theorem:finite-memory-determinacy}. Note that as $\S$ is a "strategy" for a game on an Eve-arena, each state $q\in S$ admits a unique path starting from $q$, consisting of the successive choices compatible with $\S$ in $G$.
	If $\S$ is "positional@@strategy", the statement is established.
	Otherwise, there exist two distinct memory states~$p$ and $q$ that correspond to the same position~$v$ in~$G$.
	We define two variants of~$\S$. The first, $\S_p = (S_p, \sigma_p)$, is identical but for the memory state~$q$ that is removed from $S$, and all edges in~$S$ reaching $q$ are redirected to~$p$. Note that as $p$ and $q$ are both mapped to a same vertex, this does not change the "initiality@initial path" on any "path".
	The strategy $\S_q$ is defined in a symmetrical way. 
	Note that both~$S_p$ and $S_q$ are strictly smaller than $S$.
	If either $\S_p$ or~$\S_q$ is winning, then the property is established.
	
	\emph{Case 1:} there are no paths from $p$ to $q$ and no paths from $q$ to $p$ in $S$. Let~$w\in  C^\omega$ be labelling of the unique path in~$S$ starting from~$p$, and let~$w'$ be labelling the path in~$S$ starting from~$q$.
	Then either~$w \resleqL w'$ or~$w'\resleqL w$ (since~$\resleqL$ is total). Without loss of generality, up to symmetry, we assume that $w \resleqL w'$. We show that $\S_q$ is winning. Indeed, let us consider an infinite "initial path". Either it is a path in~$S$, and in this case, it is winning since~$\S$ is winning.
	Or it consists of a "path" in~$S$ ending in~$q$ labelled by~$u\in  C^*$ followed by the path labelled by $w'$.
	Since~$uw\in L$ and $w\resleqL w'$, we obtain $uw'\in L$. The path is winning for Eve.
	
	\emph{Case 2:} there is a path from $p$ to $q$ in~$S$ and no paths from $q$ to $p$. Let~$v\in C^+$ be labelling the minimal path from~$p$ to~$q$, and $w$ be labelling the infinite path starting from~$q$. 
	Two cases can occur: either $v^\omega\resleqL w$ or $w\resleqL v^\omega$.
	First subcase: $v^\omega\resleqL w$. Let us show that $\S_q$ is winning. Let us consider an infinite "initial path" in~$S_q$. If it is a path in~$S$, then it is winning for Eve since~$\S$ is winning. Otherwise it consists of a path in~$S$ ending in~$p$, say labelled by~$u$ followed by $w$. Since~$\S$ is winning, $uvw\in L$. Therefore, since $L$ satisfies the "local preference properties", $uv^\omega \in L$ or $uw\in L$. As $v^\omega \resleqL w$, in either case, $uw\in L$ and the path is winning. 
	Otherwise, $w \resleqL v^\omega$. Let us show that $\S_p$ is winning. We observe that the only infinite path starting in $p$ is now labelled by $v^\omega$. We once more consider an infinite path in~$S_p$. If it is a path in $S$, it is immediately winning. Otherwise, it consists of a path in~$S$ ending in~$q$, say labelled by $u$, followed by the path labelled by $w$. Once again, as~$\S$ is winning, $uw \in L$. Therefore, as $w \resleqL v^\omega$, $uv^\omega \in L$, and the path is winning.
	
	\emph{Case 3:} there is no path from~$p$ to~$q$ and a path from $q$ to~$p$ in $S$. This is case 2, when the roles of~$p$ and $q$ are exchanged. 
	
	\emph{Case 4:} there is a path from~$p$ to~$q$ and a path from~$q$ to~$p$ in~$S$. 
	Let~$v$ be labelling the shortest path from~$p$ to~$q$, and $v'$ labelling the shortest path from~$q$ to~$p$.
	Up to symmetry, we assume that ${v'}^\omega \resleqL v^\omega$.
	Together with the third item of \Cref{lemma:leqslant}, we obtain that both $(vv')^\omega \resleqL v^\omega$ and $(v'v)^\omega \resleqL v^\omega$. Let us show that $\S_p$ is winning -- note that the only infinite path in $S_p$ starting in $p$ is labelled by $v^\omega$.
	Let us consider some "initial path" in~$S_p$.
	If it is a path in~$\S$ it is winning for Eve. 
	Otherwise, this means that it begins by an "initial path" in~$S$ ending in~$q$, say labelled by~$u$, followed by $v^\omega$. Since~$S$ is winning, this means that $u(v'v)^\omega\in L$. Since furthermore~$(v'v)^\omega\resleqL v^\omega$, we obtain $uv^\omega\in L$. The path is winning.	
\end{proof}

\begin{corollary}\label{corollary:difficult-direction}
	If an "$\omega$-regular language"~$L$ satisfies the three "local preference properties", then it is "Eve-positional" over all "arenas".
\end{corollary}
\begin{proof}
	Let~$L\subseteq C^\omega$ be an "$\omega$-regular language" that satisfies the three "local preference properties".
	From \Cref{lemma:posi-union-gen-hard}, it is "uniformly Eve-positional" over all  finite "Eve-only games".
	By \Cref{theorem:CO}, this means that $L$ is "Eve-positional" over all "arenas".
\end{proof}